\documentclass[11pt,a4paper]{article}

\usepackage{amsmath,amssymb,amsthm,graphicx,braket,multirow,authblk,amsthm,dcolumn,latexsym,subcaption,mathtools,cite}
\usepackage[normalem]{ulem}
\usepackage[a4paper]{geometry}
\DeclareSymbolFontAlphabet{\amsmathbb}{AMSb}%

\newtheorem{Proposition}{Proposition}[section]
\newtheorem{Corollary}{Corollary}[section]

\theoremstyle{definition}
\newtheorem{Example}{Example}

\def\oper{{\mathchoice{\rm 1\mskip-4mu l}{\rm 1\mskip-4mu l}
		{\rm 1\mskip-4.5mu l}{\rm 1\mskip-5mu l}}}
\newcommand{\hilb}{\mathcal{H}}
\newcommand{\hilbs}{\mathcal{H}_{\rm S}}
\newcommand{\hilbb}{\mathcal{H}_{\rm B}}
\newcommand{\vac}{\mathrm{vac}}
\newcommand{\e}{\mathrm{e}}
\renewcommand{\H}{\textbf{H}}
\renewcommand{\Im}{\operatorname{Im}}
\renewcommand{\Re}{\operatorname{Re}}
\newcommand{\W}[1]{W\!\left(#1\right)}
\newcommand{\Wdag}[1]{W^\dag\!\left(#1\right)}
\newcommand{\pvint}{\mathrm{PV}\!\int}
\newcommand{\ketbra}[2]{| #1 \rangle\! \langle #2 | }
\newcommand{\ii}{\mathrm{i}}
\newcommand{\tr}{\operatorname{Tr}\!}

\usepackage{hyperref}
\hypersetup{colorlinks}
\usepackage{xcolor}
\definecolor{cblue}{rgb}{0.16, 0.32, 0.75}
\definecolor{cred}{rgb}{0.7, 0.11, 0.11}
\hypersetup{%
	,linkcolor=cred
	,citecolor=cblue
	,urlcolor=black
}

\begin{document}
	
	\title{
		\textbf{Quantum regression in dephasing phenomena}
	}
	
	\author[$\hspace{0cm}$]{Davide Lonigro$^{1,2,}$\footnote{davide.lonigro@ba.infn.it}}
	\affil[$1$]{\small Dipartimento di Fisica and MECENAS, Universit\`{a} di Bari, I-70126 Bari, Italy}
	\affil[$2$]{\small INFN, Sezione di Bari, I-70126 Bari, Italy}
	
	\author[$\hspace{0cm}$]{Dariusz Chru\'sci\'nski$^{3,}$\footnote{darch@fizyka.umk.pl}}
	\affil[$3$]{\small Institute of Physics, Faculty of Physics, Astronomy and Informatics, Nicolaus Copernicus University, Grudziadzka 5/7, 87-100 Toru\'n, Poland}
	
	\maketitle
	\vspace{-0.5cm}
	
	\begin{abstract}
		We investigate the validity of quantum regression for a family of quantum Hamiltonians on a multipartite system leading to phase-damping reduced dynamics. After finding necessary and sufficient conditions for the CP-divisibility of the corresponding channel, we evaluate a hierarchy of equations equivalent to the validity of quantum regression under arbitrary interventions; in particular, we find necessary conditions for a nontrivial dephasing to be compatible with quantum regression. In this framework, we study a class of dephasing-type generalized spin-boson (GSB) models, investigating the existence of qubit-environment coupling functions that ensure the exact validity of quantum regression.
	\end{abstract}
	
	\section{Introduction}
	
	Open quantum systems attract a lot of attention both from a theoretical and an experimental point of view: nowadays, the theory of open quantum systems plays a fundamental role in many areas of natural science including chemistry, atomic and molecular physics, quantum optics, condensed matter physics, and quantum information~\cite{Breuer,Legget,Weiss,RIVAS}.
	
	A fundamental feature of such systems is quantum coherence. A resource theory of quantum coherence, based on the seminal paper~\cite{Coh1}, was formulated in recent years~\cite{Coh2,Coh3,Coh4}; moreover, appropriate measures of coherence (analog of entanglement measures) were proposed~\cite{Coh1,Coh2}. Decoherence~\cite{Car,Zurek,DEC1,DEC2} -- the loss of quantum coherence  -- has a detrimental impact on the efficiency of quantum algorithms~\cite{QIT,Palma}; consequently, the protection of a system against decoherence is an important issue of modern quantum technologies.
	
	A particular instance of quantum decoherence is {\em dephasing}, which consists in the reduction of the off-diagonal elements of the density matrix associated with the system as a result of the interaction with the environment. A typical pure dephasing process is the reduced dynamics governed by the following total system-bath Hamiltonian:	
	\begin{equation}\label{I}
		\mathbf{H} = H_{\rm S} \otimes \oper_{\rm B} + \oper_{\rm S} \otimes H_{\rm B} + \sum_{j=0}^{d-1} \ketbra{j}{j} \otimes B_j ,
	\end{equation}
	where the Hamiltonian of the system reads $H_{\rm S} = \sum_{j=0}^{d-1} \epsilon_j \ketbra{j}{j}$. Indeed, in such a case, the Hamiltonian admits a direct sum decomposition:
	\begin{equation}\label{II}
		\mathbf{H} =  \sum_{j=0}^{d-1} \ketbra{j}{j} \otimes H_j ,
	\end{equation}
	with $H_j = \epsilon_j \oper_{\rm B} + H_{\rm B} + B_j$, i.e. $\mathbf{H}$ is block-diagonal. Due to Eq.~\eqref{II}, the reduced evolution of the system, starting from a factorized state $\rho\otimes\rho_{\rm B}$, is given by $\rho_t = \Lambda_t(\rho)$, where
	\begin{equation}\label{DM}
		\Lambda_t(\rho) = {\rm Tr}_{\rm B} \left[\e^{-\ii t\mathbf{H}} (\rho \otimes \rho_{\rm B})\,\e^{\ii t\mathbf{H}} \right] = \sum_{j,\ell=0}^{d-1} \varphi_{j\ell}(t) \ketbra{j}{j} \rho \ketbra{\ell}{\ell} ,
	\end{equation}
	with the dephasing functions $\varphi_{j\ell}(t)$ being defined via	
	\begin{equation}\label{}
		\varphi_{j\ell}(t) = {\tr}\left[ \e^{-\ii tH_j} \rho_{\rm B} \,\e^{\ii tH_\ell} \right] .
	\end{equation}
	The matrix elements $\rho_{j\ell}$ of the initial state evolve according to $\rho_{jl}(t) = \varphi_{j\ell}(t) \rho_{j\ell}$, that is, the off-diagonal elements $\rho_{jl}$ are ``dephased'' via $\varphi_{j\ell}(t)$ (note that $\varphi_{jj}(t)=1$). Therefore, for a given choice of the ``block'' Hamiltonians $H_j$ and of the initial bath state $\rho_{\rm B}$, all physical properties of the dephasing process are encoded into the dephasing matrix $\left(\varphi_{j\ell}(t)\right)_{j,\ell}$. Eq.~\eqref{DM} defines a dynamical map $\{\Lambda_t\}_{t\geq 0}$: for any $t \geq 0$, the map $\Lambda_t : \mathcal{B}(\hilbs) \rightarrow\mathcal{B}(\hilbs)$ is completely positive and trace-preserving (CPTP).
	
	Here we shall be concerned with the Markovianity properties of dephasing processes. There are several, inequivalent definitions of quantum non-Markovianity (cf.\ recent reviews~\cite{NM1,NM2,NM3,NM4,Jyrki1,Jyrki2}) which, essentially, can either solely involve the properties of the dynamical map $\{\Lambda_t\}_{t\geq 0}$ or those of the full unitary evolution of the system \textit{and} the bath. Clearly, in the latter case one needs full access to the system-bath dynamics.
	
	At the level of the dynamical map $\{\Lambda_t\}_{t\geq 0}$, the evolution is usually considered to be Markovian whenever it is CP-divisible~\cite{RHP}, that is, for any $t>s$ there exists a CPTP propagator $V_{t,s}$ such that $\Lambda_t = V_{t,s} \Lambda_s$; if $V_{t,s}$ is only positive and trace-preserving, one usually calls $\{\Lambda_t\}_{t\geq 0}$ to be P-divisible (cf.~\cite{Sabrina} for a refined classification). Another concept of Markovianity based on the distinguishability of quantum states was proposed in Ref.~\cite{BLP}: one calls $\{\Lambda_t\}_{t\geq 0}$ to be Markovian whenever, for any pair of initial states $\rho_1$ and $\rho_2$, the so-called BLP property is satisfied:
	\begin{equation}\label{BLP}
		\frac{\mathrm{d}}{\mathrm{d} t} \left\| \Lambda_t(\rho_1-\rho_2)\right\|_1 \leq  0 ,
	\end{equation}
	for any $t \geq 0$, with $\|\cdot\|_1$ denoting the trace norm. In fact, P-divisibility implies Eq.~\eqref{BLP}, but the vice versa is not true~\cite{Angel,Haikka,BOGNA,PRL-Angel}. Both approaches, i.e. CP-divisibility and the BLP property~\eqref{BLP}, have been successfully applied to several systems~\cite{NM1,NM2,NM3}.
	
	However, such notions of Markovianity are not suitable as a quantum generalization of the original classical Markov condition for the family of conditional probabilities~\cite{Hanggi1,Hanggi2}, consisting in a hierarchy of conditions for the conditional probabilities of the process~\cite{Kampen}:	
	\begin{equation}\label{Markov}
		p\!\left(x_n,t_n|x_{n-1},t_{n-1};\ldots;x_1,t_1\right) = p\!\left(x_n,t_n|x_{n-1},t_{n-1}\right).
	\end{equation}
	In fact, if the dynamical map is CP-divisible, then the corresponding propagator $V_{t,s}$ does provide a generalization of the 2-point object $p(y,t|x,s)$ where `$x$' and `$y$' are the possible outcomes of two measurements performed at times `$s$' and `$t$', respectively; however, the knowledge of the map $\{\Lambda_t\}_{t\geq 0}$ does not allow to define the remaining elements from the hierarchy~\eqref{Markov}.
	
	Having access to the system-bath evolution, one may take into account a refined notion of quantum Markovianity which provides a direct generalization of the original classical concept. A proper mathematical formulation of Markov quantum stochastic process was proposed in Refs.~\cite{Lewis,Lindblad,Accardi}: in this approach, the proper quantum generalization of the classical Markov property~\eqref{Markov} consists in a hierarchy of nontrivial conditions for multi-time correlation functions known as the quantum regression formula~\cite{Lax,Zoller}. Basically, quantum regression means that all multi-time correlations of the system, derived in terms of the full system-bath evolution and of interventions performed on the system alone, can be recovered in terms of the dynamical map $\{\Lambda_t\}_{t \geq 0}$ alone~\cite{Zoller}.
	
	It is also worth pointing out that, recently, an interesting approach to quantum Markovianity was proposed (cf.~\cite{kavan1,kavan2,kavan3,Kavan-PRX}): the Markovianity of a quantum process is characterized via the factorization of the so-called quantum process tensor of the system. Again, this approach requires full access to the system-bath dynamics. Interestingly, the factorization of the process tensor is essentially equivalent to the validity of quantum regression (cf.~\cite{NM4} for a comparative analysis).
	
	The analysis of Markovianity based on the quantum regression formula was already initiated in Refs.~\cite{Francesco} and~\cite{Bassano}, where it was shown that CP-divisible maps can violate quantum regression; in fact, even the so-called Markovian semigroups $\Lambda_t = \e^{t \mathcal{L}}$, with $\mathcal{L}$ being a Gorini-Kossakowski-Lindblad-Sudarshan (GKLS) generator~\cite{GKS,L}, can violate quantum regression. Recently, we discussed the validity of quantum regression for a class of generalized spin-boson (GSB) models yielding amplitude-damping dynamics~\cite{regression} (see also the recent papers~\cite{Teretenkov,Teretenkov2,Khan,QR-PRL}). 
	
	In this paper we provide a systematic study of CP-divisibility and quantum regression for the dephasing-type quantum evolution produced by system-bath Hamiltonians in the form~\eqref{II}, starting from the qubit case and eventually investigating the general scenario. The analysis of the general case is accompanied by a thorough investigation of a concrete example of paramount importance in open quantum system theory: the dephasing-type spin-boson model and its multilevel generalizations, which again belong to the family of GSB models.
	
	The paper is organized as follows:
	\begin{itemize}
		\item in Section~\ref{SEC-II} we provide a detailed analysis of the qubit dephasing dynamics in the general case: after recalling the conditions under which CP-divisibility holds (Prop.~\ref{prop:semigroup}), we find a hierarchy of necessary and sufficient conditions for quantum regression to hold under arbitrary interventions (Prop.~\ref{prop:regression}), and we show that a nontrivial dephasing can only satisfy such conditions if the block Hamiltonians do \textit{not} commute;
		\item in Section~\ref{SEC-III} we particularize our discussion to the dephasing-type spin-boson model, showing that quantum regression, depending on the choice of the qubit-boson coupling, holds in a particular limit (Prop.~\ref{prop:flat});
		\item in Section~\ref{SEC-IV} we finally investigate the qu$d$it dephasing dynamics ($d\geq2$), again characterizing CP-divisibility (Prop.~\ref{prop:semigroup-multi}) and deriving a hierarchy of necessary and sufficient conditions for quantum regression (Prop.~\ref{prop:regression_multi}); we also study dephasing-type GSB models, and investigate choices of the coupling functions for which quantum regression holds (Props.~\ref{prop:flat_multi0}--\ref{prop:flat_multi}).
	\end{itemize}
	Final considerations are collected in Section~\ref{SEC-V}.	
	
	\section{Qubit dephasing dynamics} \label{SEC-II}
	
	\subsection{Generalities}
	
	We shall consider a self-adjoint Hamiltonian $\mathbf{H}$ on a Hilbert space $\hilb=\hilbs \otimes\hilbb$, with $\dim\hilbs=2$, given by
	\begin{equation}\label{eq:h}
		\mathbf{H} = \ketbra{0}{0}\otimes H_0+\ketbra{1}{1}\otimes H_1\simeq\begin{pmatrix}
			H_0&\\&H_1
		\end{pmatrix},
	\end{equation}
	with $\ket{0}$, $\ket{1}$ being an orthonormal basis of $\hilbs$, and $H_0,H_1$ being two (possibly) unbounded self-adjoint operators on $\hilbb$, each defined on a dense domain $\mathcal{D}(H_j)\subset\hilbb$, $j=0,1$. The unitary propagator induced by $\mathbf{H}$ on the total Hilbert space decomposes as
	\begin{equation}
		\e^{-\ii t\mathbf{H}}=\ketbra{0}{0}\otimes\e^{-\ii t H_0}+\ketbra{1}{1}\otimes\e^{-\ii tH_1}.
	\end{equation}
	Consequently, the reduced dynamics induced by $\mathbf{H}$ on the finite-dimensional space $\hilbs$, represented by the family of completely positive and trace-preserving (CPTP) maps $\Lambda_t:\mathcal{B}(\hilbs)\rightarrow\mathcal{B}(\hilbs)$ defined as in Eq.~\eqref{DM}, with $\rho_{\rm B}\in\hilbb$ being a fixed density operator of the environment, reads
	\begin{equation}
		\Lambda_t(\rho)=\sum_{j,\ell=0,1}\tr\left[\e^{-\ii tH_{j}}\rho_{\rm B}\,\e^{\ii tH_{\ell}}\right]\,\ketbra{j}{j}\rho\ketbra{\ell}{\ell},
	\end{equation}
	that is, explicitly, a dephasing structure emerges:
	\begin{equation}\label{eq:channel}
		\Lambda_t(\rho)=\begin{pmatrix}
			\rho_{00}&\rho_{01}\,\varphi(t)\\
			\rho_{10}\,\varphi(t)^*&\rho_{11}
		\end{pmatrix},
	\end{equation}
	where
	\begin{equation}
		\varphi(t)=\tr\left[\e^{-\ii tH_0}\rho_{\rm B}\,\e^{\ii tH_1}\right].
	\end{equation}
	In particular, if $\rho_{\rm B}$ is a pure state, $\rho_{\rm B}=\ketbra{\psi_{\rm B}}{\psi_{\rm B}}$, then
	\begin{equation}\label{eq:dephasing}
		\varphi(t)=\Braket{\psi_{\rm B}|\e^{\ii tH_1}\e^{-\ii tH_0}\psi_{\rm B}}=\Braket{\e^{-\ii tH_1}\psi_{\rm B}|\e^{-\ii tH_0}\psi_{\rm B}},
	\end{equation}
	so that $|\varphi(t)|^2$ coincides with the \textit{fidelity} of the two states corresponding to the evolution of $\psi_{\rm B}$ generated by $H_0$ and $H_1$. In this sense, the dephasing is a byproduct of the distinguishability between the evolution induced by the two Hamiltonians, with the dephasing being trivial in the case $H_1=H_0$.
	
	We shall refer to $\varphi(t)$ as the \textit{dephasing function} of the process. Clearly, the properties of $\Lambda_t$ are entirely dependent on the dephasing function, which, in turn, depends on the choice of the environment state $\rho_{\rm B}$ and the two block operators $H_0$, $H_1$. In any case, $\varphi(t)$ always satisfies the following properties:
	\begin{itemize}
		\item $\varphi(0)=1$, and $|\varphi(t)|\leq1$;
		\item $t\mapsto\varphi(t)$ is continuous,
	\end{itemize}
	the latter following from the strong continuity of $t\mapsto\e^{-\ii tH_0},\e^{-\ii tH_1}$. Besides, the following  properties hold:
	\begin{Proposition}\label{prop:semigroup}
		Let $\varphi(t)\neq0$ for all $t\geq0$. Then the process $t\mapsto\Lambda_t$ is invertible; besides, the following statements are equivalent:
		\begin{itemize}
			\item[(i)] $\Lambda_t$ is CP-divisible;
			\item[(ii)] $\Lambda_t$ is P-divisible;
			\item[(iii)] $t\mapsto|\varphi(t)|$ is non-increasing for all $t\geq0$.
		\end{itemize}
		Finally, $\Lambda_t$ satisfies the semigroup property $\Lambda_t=\Lambda_{t-s}\Lambda_s$ for all $t\geq s\geq0$ if and only if
		\begin{equation}\label{eq:exp}
			\varphi(t)=\e^{-\left(\ii\Omega+\frac{\gamma}{2}\right)t},\qquad t\geq0
		\end{equation}
		for some $\Omega\in\mathbb{R}$ and $\gamma\geq0$.
	\end{Proposition}
	Prop.~\ref{prop:semigroup} is a special case of Prop.~\ref{prop:semigroup-multi} in Section~\ref{SEC-IV}, which will be proven later on. Notice that, in particular, CP-divisibility and P-divisibility are \textit{equivalent} for the dephasing channel, analogously to what happens to the amplitude-damping channel and its multilevel generalization studied in Ref.~\cite{regression}. Note that, if $\Lambda_t$ defines a dynamical semigroup, then $\dot{\Lambda}_t = \mathcal{L} \Lambda_t$, with $\mathcal{L}$ being a GKLS generator given by		
	\begin{equation}\label{}
		\mathcal{L}(\rho) = -\ii \frac \Omega 2  [\sigma_z,\rho] + \frac 12 \gamma (\sigma_z \rho \sigma_z - \rho) .
	\end{equation}
	
	This proposition imposes strict conditions upon the environment state $\rho_{\rm B}$ (or, in the case of a pure state, $\psi_{\rm B}$) and the blocks $H_0$, $H_1$ in order for the corresponding dynamics to be CP-divisible. Remarkably, since $\varphi(t)$ depends on the interplay between two distinct operators, it is possible, in principle, to obtain a purely exponential decay (and therefore semigroup dynamics) with a positive Hamiltonian, $\mathbf{H}\geq0$, contrarily to what happens to the survival probability of states in a closed quantum system, which cannot decay exponentially unless the corresponding Hamiltonian has a doubly unbounded spectrum. We shall come back on this point later on.
	
	Before proceeding with the analysis of quantum regression, let us discuss some examples.
	
	\begin{Example}[Shallow-pocket Hamiltonians]\label{ex:1}
		A remarkably simple example of Hamiltonian $\mathbf{H}$ yielding a purely exponential dephasing is the shallow-pocket model (cf.~\cite{Kavan-PRX}), which we briefly discuss here in a slightly more general form. Take $H_0=-H_1\equiv \frac{1}{2}H$ in Eq.~\eqref{eq:h}; then the dephasing function~\eqref{eq:dephasing} simply reads
		\begin{equation}
			\varphi(t)=\Braket{\psi_{\rm B}|\e^{-\ii tH}\psi_{\rm B}},
		\end{equation}
		so that $\varphi(t)$ reduces to the \textit{survival amplitude} of the state $\psi_{\rm B}$ under the evolution generated by $H$. In such a case, $\varphi(t)$ is known to be a continuous function of positive type~\cite{stewart,rudin,loomis,reedsimon}, and is uniquely identified as the Fourier transform of the \textit{spectral measure} associated with $H$ and $\psi_{\rm B}$: every dephasing function of positive type can be reproduced this way. This includes purely exponential decay: $\varphi(t)$ is exponential if and only if the spectral measure is a Cauchy distribution. For example, by choosing $\hilbb=L^2(\mathbb{R})$ (i.e. the space of square-integrable functions on the real line), $H$ as the position operator on it, and
		\begin{equation}
			|\psi_{\rm B}(x)|^2=\frac{\gamma}{2\pi}\frac{1}{(x-\Omega)^2+\frac{\gamma^2}{4}},
		\end{equation}
		then $\varphi(t)$ is as in Eq.~\eqref{eq:exp}.
	\end{Example}
	
	\begin{Example}[Generalizations of shallow-pocket Hamiltonians]\label{ex:1bis}
		A generalization of the models above can be obtained by assuming $H_0=f_0(H)$ and $H_1=f_1(H)$ for two real-valued functions $f_0(x),f_1(x)$, with $H$ being a self-adjoint operator on $\hilbb$ (the case of Example~\ref{ex:1} is recovered by setting $f_0(x)=-f_1(x)=x/2$). In such a case, again we have
		\begin{equation}
			\varphi(t)=\Braket{\psi_{\rm B}|\e^{-\ii t \left(f_0(H)-f_1(H)\right)}\psi_{\rm B}}.
		\end{equation}
		Such a construction was used in Ref.~\cite{positive} to show that dephasing channels with exponential $\varphi(t)$ (thus satisfying the semigroup property) can be obtained even if $\textbf{H}$ is a positive Hamiltonian, defusing the standard argument that, for a \textit{closed} system, prohibits exponential decay at large times unless the Hamiltonian is doubly unbounded~\cite{Khalfin57,Khalfin58}. To this purpose, consider a doubly unbounded Hamiltonian $H$ and take
		\begin{equation}
			f_0(x)=\max\{0,x\},\qquad f_1(x)=-\min\{0,x\};
		\end{equation}
		then both $H_0=f_0(H)$ and $H_1=f_1(H)$ are positive Hamiltonians, and so is $\mathbf{H}$, but $f_0(H)-f_1(H)=H$ is doubly unbounded and therefore $\varphi(t)=\Braket{\psi_{\rm B}|\e^{-\ii t H}\psi_{\rm B}}$ \textit{can} decay exponentially. In particular, choosing $H$ and $\psi_{\rm B}$ as in Example~\ref{ex:1}, we obtain a purely exponential dephasing at all times (and, therefore, pure semigroup dynamics) despite $\mathbf{H}$ being positive, as pointed out in Ref.~\cite{positive}.
	\end{Example}
	
	\begin{Example}[The dephasing-type spin-boson Hamiltonian]\label{ex:2}
		Let $\hilbb$ be the symmetric Fock space associated with a boson bath whose energies cover a continuous subset of the real line. Define
		\begin{equation}
			\mathbf{H}=H_{\rm S}\otimes\oper_{\rm B}+\oper_{\rm S}\otimes H_{\rm B}+\sigma_z\otimes\left(b(f)+b^\dag(f)\right),
		\end{equation}
		with $H_{\rm S}=\omega_0\ketbra{0}{0}+\omega_1\ketbra{1}{1}$ being the free Hamiltonian of a qubit, $H_{\rm B}$ the free Hamiltonian of the boson bath, $\sigma_z=\ketbra{0}{0}-\ketbra{1}{1}$, and $b(f)$, $b^\dag(f)$ the annihilation and creation operators associated with a square-integrable function $f$, the \textit{form factor} of the model, on (a subset of) the real line. As is customary in physics, we may write them via the formal expressions
		\begin{equation}
			H_{\rm B}=\int\mathrm{d}\omega\;\omega\,b^\dag_\omega b_\omega,\qquad b(f)=\int\mathrm{d}\omega\;f(\omega)^*b_\omega,
		\end{equation}
		with $b_\omega,b_\omega^\dag$ satisfying the standard commutation rules: $[b_\omega,b_{\omega'}]=0$ and $[b_\omega,b_{\omega'}^\dag]=\delta(\omega-\omega')$. Finally, $\sigma_z=\ketbra{0}{0}-\ketbra{1}{1}$.
		
		This Hamiltonian belongs to the class of generalized spin-boson (GBS) models~\cite{arai1,arai2,arai3}, and is thus self-adjoint on $\mathcal{D}(\mathbf{H})=\hilbs\otimes\mathcal{D}(H_{\rm B})\simeq\mathcal{D}(H_{\rm B})\oplus\mathcal{D}(H_{\rm B})$; it describes a two-level system interacting with a continuous boson field in a way which, while not affecting the atom population, induces decoherence. Indeed, one immediately shows that $\mathbf{H}$ admits the decomposition~\eqref{eq:h}, with
		\begin{equation}
			H_0=\omega_0+H_{\rm B}+\left(b(f)+b^\dag(f)\right),\qquad H_1=\omega_1+H_{\rm B}-\left(b(f)+b^\dag(f)\right),
		\end{equation}
		with domain $\mathcal{D}(H_0)=\mathcal{D}(H_1)=\mathcal{D}(H_{\rm B})$. Consequently, choosing any initial state of the boson environment and tracing out the bosonic degrees of freedom, a dephasing channel is obtained.
		
		Notably, in such a case $H_0$ and $H_1$ do \textit{not} commute unless the spin-boson coupling is trivial, i.e. $f\equiv0$. As we will see later, this will prove to be crucial for quantum regression.
	\end{Example}
	
	\subsection{Hierarchy of conditions for quantum regression}
	
	We shall now investigate quantum regression for the qubit dephasing dynamics. Recall that, given an open quantum system represented by a Hamiltonian $\mathbf{H}$ on a Hilbert space $\hilb=\hilbs\otimes\hilbb$, we say that the couple\footnote{Note that, when the state $\rho_{\rm B}$ is fixed \textit{ab initio}, with a slight abuse of notation we will simply say that ``$\H$ satisfies quantum regression''.} $(\mathbf{H},\rho_{\rm B})$, with $\rho_{\rm B}$ being a state of the environment, satisfies the \textit{quantum regression hypothesis} (or simply quantum regression) if the following conditions hold~\cite{Lax}: for every system state $\rho$, every couple of families $\{X_0,\ldots,X_n\}$, $\{Y_0,\ldots,Y_n\}\subset\mathcal{B}(\hilbs)$, and all $t_n\geq t_{n-1}\geq\ldots\geq t_0\geq0$, we must have
	\begin{equation}\label{eq:regr}
		\tr_{\rm SB}\Bigl[\tilde{\mathcal{E}}_n\,\mathcal{U}_{t_n-t_{n-1}}\cdots\tilde{\mathcal{E}}_0\,\mathcal{U}_{t_0}\left(\rho\otimes\rho_{\rm B}\right)\Bigr]=\tr_{\rm S}\left[\mathcal{E}_n\,\Lambda_{t_n-t_{n-1}}\cdots\mathcal{E}_0\,\Lambda_{t_0}\left(\rho\right)\right],
	\end{equation}
	where $\mathcal{U}_t=\e^{-\ii t\mathbf{H}}(\cdot)\e^{\ii t\mathbf{H}}$, $\Lambda_t=\tr_{\rm B}\left[\mathcal{U}_t\left(\,\cdot\,\otimes\rho_{\rm B}\right)\right]$, and
	\begin{equation}
		\tilde{\mathcal{E}}_k=(X_k\otimes\oper_{\rm B})(\cdot)(Y_k\otimes\oper_{\rm B}),\qquad\mathcal{E}_k=X_k(\cdot)Y_k.
	\end{equation}
	Operationally, Eq.~\eqref{eq:regr} can be interpreted as follows: all multi-time correlation functions that describe a sequence of interventions on the system alone at times $t_0,t_1,\ldots,t_n$, with the system and the environment evolving freely between consecutive interventions, can be expressed in terms of the reduced dynamics alone. This is a highly nontrivial property: interventions can, indeed, reveal profound differences between open quantum systems whose reduced dynamics is nevertheless exactly the same.
	
	In general, Eq.~\eqref{eq:regr} must be verified for all choices of interventions and all system states $\rho$. Remarkably, for open quantum systems associated with global Hamiltonians $\mathbf{H}$ as in Eq.~\eqref{eq:h}, it is instead possible to translate Eq.~\eqref{eq:regr} in an explicit form which \textit{only} depends on $H_0$, $H_1$, and $\rho_{\rm B}$.
	\begin{Proposition}\label{prop:regression}
		Given a Hamiltonian $\H$ in the form~\eqref{eq:h} and an environment state $\rho_{\rm B}$, the couple $(\H,\rho_{\rm B})$ satisfies quantum regression if and only if, for all $t_n\geq t_{n-1}\geq\ldots\geq t_0\geq0$, the following equality holds:
		\begin{equation}\label{eq:condition}
			\tr\,\Bigl[\e^{-\ii \Delta t_n H_{j_n}}\cdots\e^{-\ii \Delta t_0 H_{j_0}}\rho_{\rm B}\,\e^{\ii\Delta t_0H_{\ell_0}}\cdots\e^{\ii \Delta t_nH_{\ell_n}}\Bigr]=\prod_{k=0}^n\tr\,\bigl[\e^{-\ii\Delta t_kH_{j_k}}\rho_{\rm B}\,\e^{\ii\Delta t_k H_{\ell_k}}\bigr]
		\end{equation}
		for all $j_0,\ell_0,\ldots,j_n,\ell_n\in\{0,1\}$, where $\Delta t_0=t_0$ and $\Delta t_k=t_k-t_{k-1}$ for $k\geq1$.
	\end{Proposition}
	\begin{proof}
		For any Hamiltonian as in Eq.~\eqref{eq:h}, we have
		\begin{eqnarray}
			\mathcal{U}_t&=&\sum_{j,\ell=0,1}\ketbra{j}{j}\cdot\ketbra{\ell}{\ell}\otimes\e^{-\ii tH_{j}}(\cdot)\e^{\ii tH_{\ell}},\\
			\Lambda_t&=&\sum_{j,\ell=0,1}\,\tr\left[\e^{-\ii tH_{j}}\rho_{\rm B}\,\e^{\ii tH_\ell}\right]\ketbra{j}{j}\cdot\ketbra{\ell}{\ell},
		\end{eqnarray}
		so that\small
		\begin{eqnarray}
			\tilde{\mathcal{E}}_n\,\mathcal{U}_{t_n-t_{n-1}}\!\cdots\!\tilde{\mathcal{E}}_0\,\mathcal{U}_{t_0}\left(\rho\otimes\rho_{\rm B}\,\right)&=&\sum_{j_n,\ell_n}\!\cdots\!\sum_{j_0,\ell_0}\Bigl[X_n\ketbra{j_n}{j_n}X_{n-1}\!\cdots\! X_0\ketbra{j_0}{j_0}\rho\ketbra{\ell_0}{\ell_0}Y_0\!\cdots\! Y_{n-1}\ketbra{\ell_n}{\ell_n}Y_n\nonumber\\
			&&\otimes \left(\e^{-\ii \Delta t_n H_{j_n}}\cdots\e^{-\ii \Delta t_0 H_{j_0}}\rho_{\rm B}\,\e^{\ii \Delta t_0H_{\ell_0}}\cdots\e^{\ii \Delta t_n H_{\ell_n}}\right)\Bigr]
		\end{eqnarray}\normalsize
		and
		\begin{eqnarray}
			\mathcal{E}_n\,\Lambda_{t_n-t_{n-1}}\cdots\mathcal{E}_0\,\Lambda_{t_0}\left(\rho\right)&=&\sum_{j_n,\ell_n}\cdots\sum_{j_0,\ell_0}\tr\left[\e^{-\ii\Delta t_0H_{j_0}}\rho_{\rm B}\,\e^{\ii\Delta t_0H_{\ell_0}}\right]\cdots\tr\left[\e^{-\ii\Delta t_nH_{j_n}}\rho_{\rm B}\,\e^{\ii\Delta t_0H_{\ell_n}}\right]\nonumber\\
			&&\times X_n\ketbra{j_n}{j_n}X_{n-1}\cdots X_0\ketbra{j_0}{j_0}\rho\ketbra{\ell_0}{\ell_0}Y_0\cdots Y_{n-1}\ketbra{\ell_n}{\ell_n}Y_n,
		\end{eqnarray}
		where all indices take the values $\{0,1\}$. Taking the trace of the two equations above, it is therefore clear that the two quantities are always equal if and only if Eq.~\eqref{eq:condition} holds.
	\end{proof}
	Some considerations are in order. First of all, notice that, defining $U_{t}^{j,\ell}=\e^{-\ii t H_{j}}(\cdot)\e^{\ii tH_\ell}$, Eq.~\eqref{eq:condition} can be written more compactly as
	\begin{equation}\label{eq:div}
		\tr\left[\mathcal{T}\!\prod_{k=0}^n U_{\Delta t_k}^{j_k,\ell_k}(\rho_{\rm B})\right]=\prod_{k=0}^n\tr\left[U_{\Delta t_k}^{j_k,\ell_k}(\rho_{\rm B})\right],
	\end{equation}
	where the $\mathcal{T}$ signals that the product of operators at the left-hand side is time-ordered.
	
	Let us examine Eq.~\eqref{eq:condition} for the smallest values of $n$. For $n=1$ the property is trivial; consequently, the first nontrivial conditions for regression are obtained for $n=2$, namely,
	\begin{equation}\label{eq:condition-n2}
		\tr\left[\e^{-\ii\Delta t_1H_{j_1}}\e^{-\ii\Delta t_0H_{j_0}}\rho_{\rm B}\,\e^{\ii\Delta t_0H_{\ell_0}}\e^{\ii\Delta t_1H_{\ell_1}}\right]=\tr\left[\e^{-\ii\Delta t_0H_{j_0}}\rho_{\rm B}\,\e^{\ii\Delta t_0H_{\ell_0}}\right]\tr\left[\e^{-\ii\Delta t_1H_{j_1}}\rho_{\rm B}\,\e^{\ii\Delta t_1H_{\ell_1}}\right],
	\end{equation}
	for all $j_0,\ell_0,j_1,\ell_1\in\{0,1\}$. These are $16$ equations; removing the trivial identities and repeated equalities, we are left with four independent constraints:\small
	\begin{eqnarray}
		\tr\left[\e^{-\ii t_1H_0}\rho_{\rm B}\,\e^{\ii t_1H_1}\right]&=&\tr\left[\e^{-\ii (t_1-t_0)H_0}\rho_{\rm B}\,\e^{\ii (t_1-t_0)H_1}\right]\tr\left[\e^{-\ii t_0H_0}\rho_{\rm B}\,\e^{\ii t_0H_1}\right];\label{eq:condition-n2-1}\\
		\tr\left[\e^{-\ii(t_1-t_0)H_0}\e^{-\ii t_0H_1}\rho_{\rm B}\,\e^{\ii t_0H_0}\e^{\ii(t_1-t_0)H_1}\right]&=&\tr\left[\e^{-\ii(t_1-t_0)H_0}\rho_{\rm B}\,\e^{\ii(t_1-t_0)H_1}\right]\tr\left[\e^{-\ii t_0H_1}\rho_{\rm B}\,\e^{\ii t_0H_0}\right];\qquad\label{eq:condition-n2-2}\\
		\tr\left[\e^{-\ii t_1H_0}\rho_{\rm B}\,\e^{\ii t_0 H_0}\e^{\ii(t_1-t_0)H_1}\right]&=&\tr\left[\e^{-\ii(t_1-t_0)H_0}\rho_{\rm B}\,\e^{\ii(t_1-t_0)H_1}\right];\label{eq:condition-n2-3}\\
		\tr\left[\e^{-\ii t_1H_1}\rho_{\rm B}\,\e^{\ii t_0 H_1}\e^{\ii(t_1-t_0)H_0}\right]&=&\tr\left[\e^{-\ii(t_1-t_0)H_1}\rho_{\rm B}\,\e^{\ii(t_1-t_0)H_0}\right].\label{eq:condition-n2-4}
	\end{eqnarray}\normalsize
	In particular, recalling the definition~\eqref{eq:dephasing} of the dephasing function, Eq.~\eqref{eq:condition-n2-1} corresponds precisely to $\varphi(t_1)=\varphi(t_1-t_0)\varphi(t_0)$ for all $t_1\geq t_0$, that is, the semigroup property. This implies that
	\begin{Corollary}
		A \textit{necessary} condition for $(\H,\rho_{\rm B})$ to satisfy quantum regression is the dephasing function $\varphi(t)$ satisfying Eq.~\eqref{eq:exp}, i.e., that the corresponding channel $\Lambda_t$ is a semigroup at all times.
	\end{Corollary}
	Therefore, for this class of models, the validity of the semigroup property $\Lambda_{t_1}=\Lambda_{t_1-t_0}\Lambda_{t_0}$ is a necessary requirement for the quantum regression hypothesis to hold. More concretely, this also means that, in order for regression to hold at least approximately at large times, $\varphi(t)$ must behave exponentially at large times.
	
	We remark that the semigroup property, while necessary, is by no means sufficient already at $n=2$; we shall indeed examine a vast class of models for which regression always fails, even with the semigroup property being satisfied.
	
	\subsection{Quantum regression and commutative dephasing}
	
	Let us analyze the case in which the operators $H_0$ and $H_1$ \textit{commute}. Two unbounded operators $H_0,H_1$ are said to (strongly) commute if their associated spectral projections commute, or, equivalently, if~\cite{reedsimon}
	\begin{equation}
		\e^{\-\ii tH_0}\e^{\ii sH_1}=\e^{\ii sH_1}\e^{\ii tH_0}\qquad\text{for all }t,s\in\mathbb{R},
	\end{equation}
	that is, if their corresponding unitary evolution groups, evaluated at any couple of times $t,s$, commute. In such a case, the order of operators in the left-hand side of Eq.~\eqref{eq:condition}, and in particular Eq.~\eqref{eq:condition-n2}, is therefore irrelevant (``the $\mathcal{T}$ can be removed'' from Eq.~\eqref{eq:div}). In particular, by using commutativity and the cyclic properties of the trace, the two conditions~\eqref{eq:condition-n2-3}--\eqref{eq:condition-n2-4} become trivial, while Eq.~\eqref{eq:condition-n2-2} can be written as
	\begin{equation}
		\tr\left[\e^{-\ii(t_1-2t_0)H_0}\rho_{\rm B}\,\e^{\ii(t_1-2t_0)H_1}\right]=\tr\left[\e^{-\ii(t_1-t_0)H_0}\rho_{\rm B}\,\e^{\ii(t_1-t_0)H_1}\right]\tr\left[\e^{-\ii t_0H_1}\rho_{\rm B}\,\e^{\ii t_0H_0}\right]\label{eq:condition-n2-2comm}
	\end{equation}
	for all $t_1\geq t_0\geq0$. But this condition alone has a fundamental consequence. Take Eq.~\eqref{eq:condition-n2-2comm} with $t_1=2t_0\equiv 2t$. Then we must have, for all $t\geq0$,
	\begin{equation}
		1=\tr\left[\e^{-\ii tH_0}\rho_{\rm B}\,\e^{\ii tH_1}\right]\tr\left[\e^{-\ii tH_1}\rho_{\rm B}\,\e^{\ii tH_0}\right]=|\varphi(t)|^2.
	\end{equation}
	We have shown the following
	\begin{Corollary}\label{coroll}
		Consider a Hamiltonian $\H$ in the form~\eqref{eq:h}, with $H_0$, $H_1$ being \textit{commuting} self-adjoint operators, and suppose that the couple $(\H,\rho_{\rm B})$ satisfies quantum regression. Then $|\varphi(t)|^2=1$, that is, the dephasing is trivial.
	\end{Corollary}
	Summing up: \textit{in the commutative case, a nontrivial dephasing is incompatible with quantum regression}, even if $\varphi(t)$ is exponential and therefore the semigroup property is satisfied at any time. This includes, in particular, the shallow-pocket model and all its generalizations introduced in Examples~\ref{ex:1}--\ref{ex:1bis}. From an operational point of view, this means that, in such models, it is possible to reveal nontrivial correlations already at $n=2$, that is, via a single intervention. In fact, in the case of the shallow-pocket model, it suffices to consider the \textit{swap} intervention:
	\begin{equation}
		\mathcal{E}=\sigma_x(\cdot)\sigma_x,
	\end{equation}
	as already observed in Refs.~\cite{kavan3,Kavan-PRX}: the swap reverses the exponential decay, revealing a truly non-Markovian feature despite the free dynamics satisfying the semigroup property (and thus, \textit{a fortiori}, CP-divisible). The statement above shows that, in fact, the failure of quantum regression in such cases can be regarded as a direct byproduct of the commutativity between the two blocks of the Hamiltonian $\mathbf{H}$.
	
	This argument does not apply when $H_0$ and $H_1$ fail to commute. This is indeed the case for the dephasing-type spin-boson model introduced in Example~\ref{ex:2}: in such a case, a nontrivial dephasing phenomenon satisfying quantum regression is therefore possible, as we will now show.
	
	\section{Dephasing-type spin-boson models} \label{SEC-III}
	
	Let us consider again the dephasing-type spin-boson model introduced in Example~\ref{ex:2}; hereafter we will set the environment state as $\rho_{\rm B}=\ketbra{\vac}{\vac}$, with $\ket{\vac}$ being the \textit{vacuum} state of the bath, formally defined via $b_\omega\ket{\vac}=0$ for all values of $\omega$. Clearly, the reduced dynamics on the spin system is a dephasing channel $\Lambda_t$ with dephasing function
	\begin{equation}\label{eq:dephasing-sb}
		\varphi(t)=\Braket{\vac|\e^{\ii tH_1}\e^{-\ii tH_0}|\vac},
	\end{equation}
	with $\varphi(t)$ crucially depending on the choice of the coupling function $f(\omega)$.
	
	\subsection{Evolution group and dephasing function}
	
	The first step in order to analyze the reduced dynamics of such models is to compute the dephasing function~\eqref{eq:dephasing-sb}. To accomplish this goal, we will mainly follow the approach of~\cite{Alicki-D} and resort to Weyl operators, whose main properties will be briefly recalled here.
	
	Given a square-integrable function $g$, the \textit{Weyl operator}
	\begin{equation}
		W(g)=\exp\left\{b^\dag(g)-b(g)\right\},
	\end{equation}
	with the exponentiation to be interpreted in the sense of the spectral theorem, is well-defined, unitary, and satisfies $W(-g)=W^\dag(g)$ (cf.\ ~\cite[Theorem X.41]{reedsimon}); besides, for all finite-particle states, i.e. all states in the form $\ket{\Psi}=b^\dag(g_1)\cdots b^\dag(g_n)\ket{\vac}$, the exponentiation can be interpreted via a series expansion. Such operators provide a representation of the Weyl algebra, that is, for all square-integrable functions $g,h$,
	\begin{equation}\label{eq:weyl_composition}
		W(g)W(h)=\exp\left\{-\ii\Im\int\mathrm{d}\omega\:g(\omega)^*h(\omega)\right\}W(g+h);
	\end{equation}
	furthermore, a direct series expansion shows
	\begin{equation}\label{eq:weyl_average}
		\Braket{\vac|W(g)|\vac}=\exp\left\{-\frac{1}{2}\int\mathrm{d}\omega\:|g(\omega)|^2\right\}.
	\end{equation}
	Finally, the equality
	\begin{equation}\label{eq:weyl-rotation}
		\e^{-\ii tH_{\rm B}}W(g)\e^{\ii tH_{\rm B}}=W(g_t),
	\end{equation}
	where $g_t(\omega)=\e^{-\ii\omega t}g(\omega)$, holds as a particular case of~\cite[Theorem X.41(e)]{reedsimon}.
	
	Let us apply this formalism to compute the dephasing function. We shall momentarily tighten our assumptions about the coupling function $f(\omega)$, and require the following two inequalities to hold:
	\begin{equation}\label{eq:constraint}
		\int\mathrm{d}\omega\;\frac{|f(\omega)|^2}{\omega}<\infty,\qquad 	\int\mathrm{d}\omega\;\frac{|f(\omega)|^2}{\omega^2}<\infty.
	\end{equation}
	The second one allows us to define the Weyl operator $\W{\frac{f}{\omega}}$. Now, a direct computation shows that the equalities
	\begin{eqnarray}
		\W{\frac{f}{\omega}}H_0\Wdag{\frac{f}{\omega}}&=&\tilde{\omega}_0+H_{\rm B};\\
		\Wdag{\frac{f}{\omega}}H_1\W{\frac{f}{\omega}}&=&\tilde{\omega}_1+H_{\rm B},
	\end{eqnarray}
	where
	\begin{equation}
		\tilde{\omega}_{0}=\omega_{0}-\int\mathrm{d}\omega\;\frac{|f(\omega)|^2}{\omega},\qquad	\tilde{\omega}_{1}=\omega_{1}-\int\mathrm{d}\omega\;\frac{|f(\omega)|^2}{\omega},
	\end{equation}
	hold on a dense subspace of $\hilbb$: both $H_0$ and $H_1$ are thus unitarily equivalent, up to an additional shift to the energies $\omega_0,\omega_1$, to their decoupled counterparts. Consequently, we have
	\begin{eqnarray}
		\label{eq:evol_h0}\e^{-\ii tH_{0}}&=&\e^{-\ii t\tilde{\omega}_{0}}\,\W{\frac{f}{\omega}}\e^{-\ii tH_{\rm B}}\Wdag{\frac{f}{\omega}},\\
		\label{eq:evol_h1}\e^{-\ii tH_{1}}&=&\e^{-\ii t\tilde{\omega}_{1}}\,\Wdag{\frac{f}{\omega}}\e^{-\ii tH_{\rm B}}\W{\frac{f}{\omega}}.
	\end{eqnarray}
	These equations, together with the general properties~\eqref{eq:weyl_composition}--\eqref{eq:weyl-rotation} of Weyl operators, will enable us to compute the dephasing function $\varphi(t)$.
	\begin{Proposition}\label{prop:dephasing}
		The dephasing function $\varphi(t)$ in Eq.~\eqref{eq:dephasing-sb} is given by
		\begin{equation}
			\varphi(t)=\e^{-\ii(\omega_0-\omega_1)t}\exp\left\{-4\int\mathrm{d}\omega\;\frac{|f(\omega)|^2}{\omega^2}(1-\cos\omega t)\right\}.
		\end{equation}
	\end{Proposition}
	\begin{proof}
		We have
		\begin{eqnarray}\label{eq:varphi}
			\varphi(t)&=&\e^{-\ii(\omega_0-\omega_1)t}\Braket{\vac\bigg|\Wdag{\frac{f}{\omega}}\e^{\ii t H_{\rm B}}\W{\frac{f}{\omega}}\W{\frac{f}{\omega}}\e^{-\ii tH_{\rm B}}\Wdag{\frac{f}{\omega}}\bigg|\vac}\nonumber\\
			&=&\e^{-\ii(\omega_0-\omega_1)t}\Braket{\vac\bigg|\Wdag{\frac{f}{\omega}}\e^{\ii t H_{\rm B}}\W{\frac{2f}{\omega}}\e^{-\ii tH_{\rm B}}\Wdag{\frac{f}{\omega}}\bigg|\vac}\nonumber\\
			&=&\e^{-\ii(\omega_0-\omega_1)t}\Braket{\vac\bigg|\Wdag{\frac{f}{\omega}}\W{\frac{2f}{\omega}\e^{\ii\omega t}}\Wdag{\frac{f}{\omega}}\bigg|\vac}\nonumber\\
			&=&\e^{-\ii(\omega_0-\omega_1)t}\Braket{\vac\bigg|\W{\frac{2f}{\omega}(\e^{\ii\omega t}-1)}\bigg|\vac}\nonumber\\
			&=&\e^{-\ii(\omega_0-\omega_1)t}\exp\left\{-\int\mathrm{d}\omega\;\frac{2|f(\omega)|^2}{\omega^2}\left|\e^{\ii\omega t}-1\right|^2\right\}\nonumber\\
			&=&\e^{-\ii(\omega_0-\omega_1)t}\exp\left\{-4\int\mathrm{d}\omega\;\frac{|f(\omega)|^2}{\omega^2}(1-\cos\omega t)\right\},
		\end{eqnarray}
		where, in order, we have expressed the evolution groups $\e^{-\ii tH_0}$, $\e^{\ii tH_1}$ via Eqs.~\eqref{eq:evol_h0}--\eqref{eq:evol_h1} (notice that $\tilde{\omega}_0-\tilde{\omega}_1=\omega_0-\omega_1$), applied Eq.~\eqref{eq:weyl_composition}, then Eq.~\eqref{eq:weyl-rotation} and again~\eqref{eq:weyl_composition}, and finally used Eq.~\eqref{eq:weyl_average}.
	\end{proof}
	This result was obtained by assuming $f(\omega)$ to be a square-integrable function also satisfying the constraints~\eqref{eq:constraint}, which imply, in particular, that the function must vanish sufficiently quickly both at $\omega\to0$ and at $|\omega|\to\infty$. However, the final result makes sense for a far larger class of coupling functions because of the presence of the factor $(1-\cos\omega t)/\omega^2$, which is $\mathcal{O}(1)$ at $\omega\to0$ and $\mathcal{O}(\omega^{-2})$ at $|\omega|\to\infty$.
	
	Consequently, every continuous coupling function $f(\omega)$ which is $\mathcal{O}(1)$ at large values of $\omega$ yields a finite dephasing function $\varphi(t)$. More precisely, given any such function, there exists a well-defined dephasing channel $\Lambda_t$ which can be obtained as the limiting reduced dynamics of a family of dephasing-type spin-boson models, obtained for instance by taking an UV and an IR cutoff:
	\begin{equation}
		f^N(\omega)=\begin{cases}
			f(\omega),&\frac{1}{N}\leq|\omega|\leq N;\\
			0,&|\omega|<\frac{1}{N},\;|\omega|>N,
		\end{cases}
	\end{equation}
	and then performing the limit $N\to\infty$ in Eq.~\eqref{eq:varphi}. We stress that, since the integral in Eq.~\eqref{eq:varphi} converges, the result is \textit{independent} of the specific choice of cutoff. This enables one to take into account possibly singular (i.e. non-normalizable) form factors (see~\cite{gsb} for an \textit{ab initio} approach to such form factors).
	
	As a particular case, it is possible to find two distinct choices of coupling such that $\varphi(t)$ is a purely exponential function at all times (and thus, by Prop.~\ref{prop:semigroup}, the dephasing channel satisfies the semigroup property at all times). One of these is a flat (white) coupling on all reals: $|f(\omega)|^2=\mathrm{const.}$, $-\infty<\omega<\infty$. Indeed, via simple arguments of complex analysis (see Eq.~\eqref{eq:complexint_result3} in the appendix), we get
	\begin{equation}\label{eq:int}
		\int_{-\infty}^{\infty}\mathrm{d}\omega\;\frac{1}{\omega^2}(1-\cos\omega t)=\Re\,\pvint_{-\infty}^{\infty}\mathrm{d}\omega\;\frac{1}{\omega^2}(1-\e^{\ii\omega t})=\pi|t|,
	\end{equation}
	finally implying that, by choosing $|f(\omega)|^2=\frac{\gamma}{8\pi}$ for some $\gamma>0$, we get
	\begin{equation}\label{eq:semi}
		\varphi(t)=\e^{-\ii t(\omega_{0}-\omega_{1})}\e^{-\gamma |t|/2},\qquad|\varphi(t)|^2=\e^{-\gamma |t|},
	\end{equation}
	and thus semigroup dynamics for $t\geq0$.
	
	Another possible coupling function yielding semigroup dynamics is a flat coupling on all \textit{positive} energies, $|f(\omega)|^2=\mathrm{const.}$, $0\leq\omega<\infty$. Indeed, by symmetry, we simply have
	\begin{equation}
		\int_0^\infty\mathrm{d}\omega\;\frac{1}{\omega^2}(1-\cos\omega t)=\frac{1}{2}\int_{-\infty}^{\infty}\mathrm{d}\omega\;\frac{1}{\omega^2}(1-\cos\omega t)=\frac{\pi|t|}{2},
	\end{equation}
	implying that a dephasing function as in Eq.~\eqref{eq:semi} can be also obtained by choosing $|f(\omega)|^2=\frac{\gamma}{4\pi}$ on positive reals.
	
	We point out that, while these choices of coupling may be considered unphysical, the corresponding results are indicative of what would be obtained in more realistic scenarios: we can expect an exponential dephasing in the regime in which the spin-boson interaction is ``approximately flat'' in the energy regime of interest. Besides, it should be possible to find (infinitely many) other choices of coupling that cause the semigroup property to be satisfied up to a finite time, similarly to what happens for the spin-boson model yielding amplitude-damping dynamics~\cite{hidden1,hidden2}.
	
	\subsection{Quantum regression for the dephasing-type spin-boson model}
	
	As previously discussed, the semigroup property does \textit{not} automatically ensure the validity of quantum regression, with simple counterexamples having been discussed before. We will now show that, in fact, the limiting case of dephasing-type spin-boson model with flat coupling \textit{does} satisfy quantum regression. 
	
	First of all, the condition~\eqref{eq:condition} which, by Prop.~\ref{prop:regression}, is equivalent to the validity of quantum regression, reads
	\begin{equation}\label{eq:condition-pd}
		\Braket{\vac\!|\e^{\ii\Delta t_0H_{\ell_0}}\dots\e^{\ii\Delta t_nH_{\ell_n}}\,\e^{-\ii\Delta t_nH_{j_n}}\dots\e^{-\ii\Delta t_0H_{j_0}}|\vac}=\prod_{k=0}^n\Braket{\vac|\e^{\ii\Delta t_kH_{\ell_k}}\e^{-\ii\Delta t_k H_{j_k}}|\!\vac}
	\end{equation}
	for all $j_0,\ell_0,\ldots,j_n,\ell_n\in\{0,1\}$.
	
	\begin{Proposition}\label{prop:flat}
		The following facts hold:
		\begin{enumerate}
			\item [(i)] The phase-damping spin-boson model with $|f(\omega)|^2=\mathrm{const.}$ on $-\infty<\omega<\infty$ satisfies quantum regression, i.e. Eq.~\eqref{eq:condition-pd} holds;
			\item [(ii)] 	The phase-damping spin-boson model with $|f(\omega)|^2=\mathrm{const.}$ on $0\leq\omega<\infty$ violates quantum regression. However, it satisfies the weaker condition\small
			\begin{equation}\label{eq:condition-pd-weaker}
				\left|\Braket{\vac\!|\e^{\ii\Delta t_0H_{\ell_0}}\dots\e^{\ii\Delta t_nH_{\ell_n}}\,\e^{-\ii\Delta t_nH_{j_n}}\dots\e^{-\ii\Delta t_0H_{j_0}}|\vac}\right|=\left|\prod_{k=0}^n\Braket{\vac|\e^{\ii\Delta t_kH_{\ell_k}}\e^{-\ii\Delta t_k H_{j_k}}|\!\vac}\right|,
			\end{equation}\normalsize
			that is, all quantum regression conditions are met up to a phase term.
		\end{enumerate}	
	\end{Proposition}
	
	\begin{proof} By Prop.~\ref{prop:regression}, quantum regression holds iff Eq.~\eqref{eq:condition-pd} holds for any $j_0,\ell_0,\ldots,j_n,\ell_n\in\{0,1\}$. Define
		\begin{equation}
			j\in\{0,1\}\mapsto s_j=\begin{cases}
				+1,&j=0;\\
				-1,&j=1.
			\end{cases}
		\end{equation}
		Now, by Eq.~\eqref{eq:varphi}, and setting again $|f(\omega)|^2=\frac{\gamma}{8\pi}$ for $\gamma>0$, we have\small
		\begin{eqnarray}
			\Braket{\vac|\e^{\ii \Delta t_kH_{\ell_k}}\e^{-\ii \Delta t_k H_{j_k}}|\vac}&=&\begin{dcases}
				\e^{-\ii \Delta t_k(\omega_{j_k}-\omega_{\ell_k})}\!\exp\left\{\!-4\int_{-\infty}^{\infty}\mathrm{d}\omega\:\frac{|f(\omega)|^2}{\omega^2}(1-\cos\omega \Delta t_k)\right\},&\!j_k\neq \ell_k\\
				1,&\!j_k=\ell_k
			\end{dcases}\nonumber\\
			&=&\e^{-\ii \Delta t_k(\omega_{j_k}-\omega_{\ell_k})}\exp\left\{-(s_{j_k}-s_{\ell_k})^2\int_{-\infty}^{\infty}\mathrm{d}\omega\;\frac{|f(\omega)|^2}{\omega^2}(1-\cos\omega \Delta t_k)\right\}\nonumber\\
			&=&\e^{-\ii \Delta t_k(\omega_{j_k}-\omega_{\ell_k})}\exp\left\{-(s_{j_k}-s_{\ell_k})^2\frac{\gamma\Delta t_k}{8}\right\},\nonumber
		\end{eqnarray}\normalsize
		where we used Eq.~\eqref{eq:int}, thus implying
		\begin{equation}\label{eq:regreq1}
			\prod_{k=0}^n\Braket{\vac|\e^{\ii \Delta t_kH_{\ell_k}}\e^{-\ii \Delta t_k H_{j_k}}|\vac}=\left(\prod_{k=0}^n\e^{-\ii \Delta t_k(\omega_{j_k}-\omega_{\ell_k})}\right)\exp\left\{-\frac{\gamma}{8}\sum_{k=0}^n(s_{j_k}-s_{\ell_k})^2\Delta t_k\right\}.
		\end{equation}
		Besides, using repeatedly Eqs.~\eqref{eq:evol_h0}--\eqref{eq:evol_h1} and the composition properties of Weyl operators, the left-hand side of Eq.~\eqref{eq:condition-pd} reads
		\begin{eqnarray}\label{eq:regreq2}
			&&	\Braket{\vac|\e^{\ii t_0H_{\ell_0}}\cdots\e^{\ii (t_n-t_{n-1})H_{\ell_n}}\e^{-\ii (t_n-t_{n-1}) H_{j_n}}\cdots\e^{-\ii t_0 H_{j_0}}|\vac}\nonumber\\
			&=&\left(\prod_{k=0}^n\e^{-\ii\Delta t_k(\omega_{j_k}-\omega_{\ell_k})}\right)\Braket{\vac\Bigg|\W{\frac{f}{\omega}\sum_{k=0}^n(s_{j_k}-s_{\ell_k})\left(\e^{\ii\omega t_k}-\e^{\ii\omega t_{k-1}}\right)}\Bigg|\vac}
		\end{eqnarray}
		and
		\begin{eqnarray}
			&&	\Braket{\vac\Bigg|\W{\frac{f}{\omega}\sum_{k=0}^n(s_{j_k}-s_{\ell_k})\left(\e^{\ii\omega t_k}-\e^{\ii\omega t_{k-1}}\right)}\Bigg|\vac}\nonumber\\
			&=&\exp\left\{-\frac{1}{2}\int_{-\infty}^{\infty}\mathrm{d}\omega\,\frac{|f(\omega)|^2}{\omega^2}\left|\sum_{k=0}^n(s_{j_k}-s_{\ell_k})\left(\e^{\ii\omega t_k}-\e^{\ii\omega t_{k-1}}\right)\right|^2\right\}\nonumber\\
			&=&\exp\left\{-\frac{1}{2}\int_{-\infty}^{\infty}\mathrm{d}\omega\,\frac{|f(\omega)|^2}{\omega^2}\sum_{k,h=0}^n(s_{j_k}-s_{\ell_k})(s_{j_{h}}-s_{\ell_{h}})\left(\e^{\ii\omega t_k}-\e^{\ii\omega t_{k-1}}\right)\left(\e^{-\ii\omega t_\ell}-\e^{-\ii\omega t_{\ell-1}}\right)\right\}\nonumber\\			&=&\exp\left\{-\frac{\gamma}{16\pi}\int_{-\infty}^{\infty}\mathrm{d}\omega\,\frac{1}{\omega^2}\sum_{k,h=0}^n(s_{j_k}-s_{\ell_k})(s_{j_{h}}-s_{\ell_{h}})\left(\e^{\ii\omega t_k}-\e^{\ii\omega t_{k-1}}\right)\left(\e^{-\ii\omega t_\ell}-\e^{-\ii\omega t_{\ell-1}}\right)\right\},\nonumber\\
		\end{eqnarray}\normalsize
		with all additional phase terms due to the combination of Weyl operators vanishing by symmetry. We are thus left with the problem of evaluating the integral\footnote{The first integral is well-defined since the numerator is $\mathcal{O}(\omega^2)$ around $\omega=0$ provided that $t_k\neq t_{k-1}$ and $t_\ell\neq t_{\ell-1}$; we are thus free to compute it as a principal value integral, and decompose it as the sum of two principal value integrals (which, instead, would not be well-defined without a principal value prescription).}
		\begin{eqnarray}
			\int_{-\infty}^{\infty}\mathrm{d}\omega\;\frac{\left(\e^{\ii\omega t_k}-\e^{\ii\omega t_{k-1}}\right)\left(\e^{-\ii\omega t_h}-\e^{-\ii\omega t_{h-1}}\right)}{\omega^2}&=&	\pvint_{-\infty}^{\infty}\mathrm{d}\omega\;\frac{\e^{\ii\omega (t_k-t_h)}-\e^{\ii\omega (t_{k-1}-t_h)}}{\omega^2}\nonumber\\
			&&-\pvint_{-\infty}^{\infty}\mathrm{d}\omega\;\frac{\e^{\ii\omega (t_k-t_{h-1})}-\e^{\ii\omega (t_{k-1}-t_{h-1})}}{\omega^2};\nonumber\\
		\end{eqnarray}
		but we have (see Eqs.~\eqref{eq:complexint_result1}--\eqref{eq:complexint_result2} in the appendix)
		\begin{eqnarray}
			\pvint_{-\infty}^{\infty}\mathrm{d}\omega\;\frac{\e^{\ii\omega (t_k-t_h)}-\e^{\ii\omega (t_{k-1}-t_h)}}{\omega^2}&=&\begin{cases}
				-\pi\left(t_k-t_{k-1}\right),&h\leq k-1;\\
				\pi\left(t_k-t_{k-1}\right),&h\geq k,
			\end{cases}\\
			\pvint_{-\infty}^{\infty}\mathrm{d}\omega\;\frac{\e^{\ii\omega (t_k-t_{h-1})}-\e^{\ii\omega (t_{k-1}-t_{h-1})}}{\omega^2}&=&\begin{cases}
				-\pi\left(t_k-t_{k-1}\right),&h\leq k;\\
				\pi\left(t_k-t_{k-1}\right),&h\geq k+1,
			\end{cases}
		\end{eqnarray}
		thus implying
		\begin{equation}
			\int_{-\infty}^{\infty}\mathrm{d}\omega\;\frac{\left(\e^{\ii\omega t_k}-\e^{\ii\omega t_{k-1}}\right)\left(\e^{-\ii\omega t_h}-\e^{-\ii\omega t_{h-1}}\right)}{\omega^2}=2\pi\left(t_k-t_{k-1}\right)\delta_{kh},
		\end{equation}
		with $\delta_{kh}$ being the Kronecker delta, and therefore
		\begin{eqnarray}\label{eq:regreq3}
			&&\sum_{k,h=0}^n(s_{j_k}-s_{\ell_k})(s_{j_h}-s_{\ell_h})\int_{-\infty}^{\infty}\mathrm{d}\omega\;\frac{1}{\omega^2}\left(\e^{\ii\omega t_k}-\e^{\ii\omega t_{k-1}}\right)\left(\e^{-\ii\omega t_h}-\e^{-\ii\omega t_{h-1}}\right)\nonumber\\&=&2\pi	\sum_{k=0}^n(s_{j_k}-s_{\ell_k})^2\Delta t_k.
		\end{eqnarray}
		Inserting Eq.~\eqref{eq:regreq3} in Eq.~\eqref{eq:regreq2}, the quantities in Eqs.~\eqref{eq:regreq1}--\eqref{eq:regreq2} are finally shown to be equal. This finally proves (i). 
		
		Identical calculations can be performed when $|f(\omega)|^2$ is constant on positive energies, with the only difference being the fact that the additional phase terms in Eq.~\eqref{eq:regreq2} do \textit{not} generally vanish. For example, the following equality holds:
		\begin{eqnarray}
			&&\Braket{\vac|\e^{\ii t_0H_0}\e^{\ii(t_1-t_0)H_1}\e^{-\ii t_1H_0}|\vac}\nonumber\\
			&=&\exp\left\{\ii\int\mathrm{d}\omega\;\frac{|f(\omega)|^2}{\omega^2}\left[\sin\omega t_0(1-\cos\omega t_1)-\sin\omega t_1(1-\cos\omega t_0)\right]\right\}\varphi(t_1-t_0);
		\end{eqnarray}
		Eq.~\eqref{eq:condition-n2-3} would require the integral in the additional phase term to vanish for all $t_1\geq t_0\geq0$, which does not happen if $|f(\omega)|^2=\text{const.}$ on the half-line. (ii) is proven.		
	\end{proof}
	
	\section{Multilevel scenario} \label{SEC-IV}
	
	\subsection{Generalities, CP-divisibility, and quantum regression}
	To conclude this work, let us extend our analysis to the multilevel (qu$d$it) case. Consider a self-adjoint Hamiltonian $\mathbf{H}$ on a Hilbert space $\hilb=\hilbs\otimes\hilbb$, where now we fix $\dim\hilbs=d\geq2$, having the following decomposition:
	\begin{equation}\label{eq:h_multi}
		\mathbf{H}=\sum_{j=0}^{d-1}\ketbra{j}{j}\otimes H_j\simeq\begin{pmatrix}
			H_0 &     &     &  \\
			& H_1 &     &  \\
			&     & H_2 &  \\
			&     &     &\ddots
		\end{pmatrix},
	\end{equation}
	again with $\{\ket{j}\}_{j=0,\dots,d-1}$ being an orthonormal base of $\hilbs$, and $H_0,H_1,\dots,H_{d-1}$ being self-adjoint operators. Following analogous computations as in the two-level scenario, one immediately shows that the family of CPTP maps $\Lambda_t:\mathcal{B}(\hilbs)\rightarrow\mathcal{B}(\hilbs)$ representing the reduced dynamics induced by $\mathbf{H}$ on the $d$-dimensional space $\hilbs$ reads as in Eq.~\eqref{DM}:
	\begin{equation}
		\Lambda_t(\rho)=\sum_{j,\ell=0}^{d-1}\tr\left[\e^{-\ii tH_{j}}\rho_{\rm B}\,\e^{\ii tH_{\ell}}\right]\,\ketbra{j}{j}\rho\ketbra{\ell}{\ell},
	\end{equation}
	that is, defining for all $j,\ell=0,\dots,d-1$
	\begin{equation}\label{eq:dephasing_multi}
		\varphi_{j\ell}(t)=\tr\left[\e^{-\ii tH_j}\rho_{\rm B}\,\e^{\ii tH_{\ell}}\right],
	\end{equation}
	we obtain
	\begin{equation}
		\Lambda_t(\rho)=
		\begin{pmatrix}
			\rho_{00}&	\rho_{01}\,\varphi_{01}(t)&	\rho_{02}\,\varphi_{02}(t)& \ldots \\
			\rho_{10}\,\varphi_{10}(t)&	\rho_{11}&	\rho_{12}\,\varphi_{12}(t)&\ldots  \\
			\rho_{20}\,\varphi_{20}(t)&	\rho_{21}\,\varphi_{21}(t)&	\rho_{22}&\ldots  \\
			\vdots&	\vdots&	\vdots & \ddots\\
		\end{pmatrix},
	\end{equation}
	where we have used the obvious equality $\varphi_{jj}(t)=1$. This is an immediate generalization of the qubit dephasing channel in Eq.~\eqref{eq:channel}, which can be written more compactly as
	\begin{equation}
		\Lambda_t(\rho)= \Phi(t) \circ \rho ,
	\end{equation}
	with $\Phi(t)=\left(\varphi_{j\ell}(t)\right)_{j,\ell}$, and $\circ$ denoting the entrywise (Hadamard) product between two equal-sized matrices.
	
	Since $\varphi_{jj}(t)=1$ and $\varphi_{\ell j}(t)=\varphi_{j\ell}(t)^*$, this process depends on $d(d-1)/2$ independent dephasing functions. As in the qubit case, each dephasing function $\varphi_{j\ell}(t)$ satisfies the following properties:
	\begin{itemize}
		\item $\varphi_{j\ell}(0)=1$, and $|\varphi_{j\ell}(t)|\leq1$;
		\item $t\mapsto\varphi_{j\ell}(t)$ is continuous,
	\end{itemize}
	and a multilevel counterpart of Prop.~\ref{prop:semigroup} holds.
	
	\begin{Proposition}\label{prop:semigroup-multi}
		Let $\varphi_{j\ell}(t)\neq0$ for all $j\neq\ell=0,\dots,d-1$, and all $t\geq0$. Then the process $t\mapsto\Lambda_t$ is invertible; besides, the following statements are equivalent:
		\begin{itemize}
			\item[(i)] $\Lambda_t$ is CP-divisible;
			\item[(ii)] $\Lambda_t$ is P-divisible;
			\item[(iii)] for all $t\geq s\geq0$, the matrix
			\begin{equation}\label{eq:matrix}
				\Phi(t)\circ\Phi(s)^{\circ -1}=
				\begin{pmatrix}
					1&\varphi_{01}(t)/\varphi_{01}(s)&\varphi_{02}(t)/\varphi_{02}(s)&\cdots\,\\
					{\varphi_{10}(t)/\varphi_{10}(s)} & 1 & \varphi_{12}(t)/\varphi_{12}(s) & \cdots\,\\
					{\varphi_{20}(t)/\varphi_{20}(s)} & {\varphi_{21}(t)/\varphi_{21}(s)} & 1 & \cdots\,\\
					\vdots&\vdots&\vdots&\ddots
				\end{pmatrix}
			\end{equation}
			is positive semidefinite.
		\end{itemize}
		In particular, a \textit{necessary} condition for CP-divisibility is that each function $t\mapsto|\varphi_{j\ell}(t)|$ is monotonically non-increasing. The condition is also sufficient for $n=2$.
		
		Finally, $\Lambda_t$ satisfies the semigroup property $\Lambda_t=\Lambda_{t-s}\Lambda_s$ for all $t\geq s\geq0$ if and only if, for all $j\neq\ell=0,\dots,d-1$,
		\begin{equation}\label{eq:exp2}
			\varphi_{j\ell}(t)=\e^{-\left(\ii\Omega_{j\ell}+\frac{\gamma_{j\ell}}{2}\right)t},\qquad t\geq0
		\end{equation}
		for some $\Omega_{j\ell}\in\mathbb{R}$ and $\gamma_{j\ell}\geq0$.
	\end{Proposition}
	\begin{proof}
		Given any matrix $A=\left(a_{j\ell}\right)_{j,\ell=0,\dots,d-1}$, consider the map
		\begin{equation}\label{eq:kraus2}
			\rho\in\mathcal{B}(\hilbs)\mapsto \Pi_A(\rho)=A\circ\rho=\sum_{j,\ell=0}^{d-1} a_{j\ell}\ketbra{j}{j}\rho\ketbra{\ell}{\ell}.
		\end{equation}
		Clearly, it satisfies the following composition law:
		\begin{equation}
			\Pi_A\Pi_B=\Pi_{A\circ B}
		\end{equation}
		and, in particular, $\Pi_A$ is invertible if and only if $A$ admits an Hadamard inverse $A^{\circ-1}$, with $\Pi^{-1}_A=\Pi_{A^{\circ-1}}$. Besides, we can easily show that the following statements are equivalent:
		\begin{itemize}
			\item[(i)] $\Pi_A$ is completely positive;
			\item[(ii)] $\Pi_A$ is positive;
			\item[(iii)] $A$ is positive semidefinite.
		\end{itemize}
		Indeed, since the right-hand side of Eq.~\eqref{eq:kraus2} corresponds to a (generally non-diagonal) Kraus representation of $\Pi_A$, necessarily $\Pi_A$ is completely positive if and only $A$ is a positive semidefinite matrix, so that (i)$\iff$(iii). Obviously (i)$\implies$(ii); to complete the proof, we need to show show (ii)$\implies$(iii). Suppose that $A$ is not positive semidefinite; consider the positive semidefinite matrix
		\begin{equation}
			X=\begin{pmatrix}
				1&1&\cdots&1\\
				1&1&\cdots&1\\\
				\vdots&\vdots&\ddots&\vdots\\
				1&1&\cdots&1
			\end{pmatrix}.
		\end{equation}
		Then obviously $\Pi_A(X)=A$, which is not positive semidefinite; therefore, $\Pi_A$ is not a positive channel. This proves the aforementioned equivalence.
		
		Now, clearly the multilevel dephasing process $t\mapsto\Lambda_t$ corresponds to $\Lambda_t=\Pi_{\Phi(t)}$ as defined in Eq.~\eqref{eq:kraus2}; consequently, the process is invertible whenever the matrix $\Phi(t)$ admits an Hadamard inverse, that is, $\varphi_{j\ell}(t)\neq0$ for all $t\geq0$. If so, we have $\Lambda_t=W_{t,s}\Lambda_s$, with
		\begin{equation}
			W_{t,s}=\Lambda_t\Lambda_s^{-1}=\Pi_{\Phi(t)}\Pi_{\Phi(s)^{\circ-1}}=\Pi_{\Phi(t)\circ\Phi(s)^{\circ-1}}.
		\end{equation}
		Consequently, the process is CP-divisible, or equivalently P-divisible, if and only if the matrix~\eqref{eq:matrix} is positive definite. 	In particular, necessarily all its entries must have modulus not greater than one, i.e.
		\begin{equation}\label{eq:mono}
			\left|\varphi_{j\ell}(t)\varphi_{j\ell}(s)^{-1}\right|\leq1\implies \left|\varphi_{j\ell}(t)|\leq|\varphi_{j\ell}(s)\right|\qquad\forall t\geq s\geq0;
		\end{equation}
		indeed, if this condition fails for some element, say $j=0,\ell=1$, and some $t_0\geq s_0\geq0$, then the matrix will admit a negative principal minor:
		\begin{equation}
			\left|\begin{array}{cc}
				1&\varphi_{01}(t_0)\varphi_{01}(s_0)^{-1}\\
				\left(\varphi_{01}(t_0)\varphi_{01}(s_0)^{-1}\right)^*&1
			\end{array}\right|=1-|\varphi_{10}(t_0)|^2|\varphi_{10}(s_0)^{-1}|^2<0,
		\end{equation}
		and thus CP-divisibility fails. In the particular case $n=2$, the condition~\eqref{eq:mono} is clearly sufficient as well.
		
		To complete the proof, suppose that $t\mapsto\Lambda_t$ satisfies the semigroup property $\Lambda_t=\Lambda_{t-s}\Lambda_s$ for all $t\geq s\geq0$. This clearly holds if and only if
		\begin{equation}
			\Phi(t-s)\circ\Phi(s)=\Phi(t)\implies\varphi_{j\ell}(t-s)\varphi_{j\ell}(s)=\varphi_{j\ell}(t),\qquad j,\ell=0,\dots,n,
		\end{equation}
		and thus, since all functions $t\mapsto\varphi_{j\ell}(t)$ are continuous and must satisfy $|\varphi_{j\ell}(t)|\leq1$, if and only if each of them is an exponentially decaying function.
	\end{proof}
	We stress that, in the general qu$d$it case, the monotonicity condition on $t\mapsto|\varphi_{j\ell}(t)|$ is only a necessary condition for CP-divisibility, differently from what happens in the qubit case. The monotonicity condition for $t\mapsto|\varphi_{j\ell}(t)|$ only ensures the nonnegativity of all principal values of order $2$, but does not ensure nonnegativity for principal minors of order $3,4,\dots,d$. We also remark that the complete positivity of $\Lambda_t$, which is guaranteed \textit{a priori} by its very definition~\eqref{DM}, can be also checked explicitly by noticing that $\Phi(t)$ is indeed a positive semidefinite matrix. Indeed, by writing Eq.~\eqref{eq:dephasing_multi} as
	\begin{equation}
		\varphi_{j\ell}(t)=\tr\left[\left(\e^{-\ii tH_\ell}\sqrt{\rho_{\rm B}}\right)^\dag\left(\e^{-\ii tH_j}\sqrt{\rho_{\rm B}}\right)\right],
	\end{equation}
	that is, as the Hilbert-Schmidt product between $\e^{-\ii tH_\ell}\sqrt{\rho_{\rm B}}$ and $\e^{-\ii tH_j}\sqrt{\rho_{\rm B}}$, clearly $\Phi(t)$ is the Gram matrix associated with the family $\{\e^{-\ii tH_j}\sqrt{\rho_{\rm B}}\}_{j=0,\dots,d-1}$, thus being positive semidefinite. Finally, notice that a family of Hadamard-type semigroups was also studied in Ref.~\cite{FABIO}.
	
	An interesting consequence of the previous proposition is the following. Recall that, given a density operator $\rho$ and a fixed orthonormal basis $\mathcal{B}=\{e_0,\ldots,e_{d-1}\}$, the \textit{coherence} of $\rho$ with respect to the basis $\mathcal{B}$ can be quantified via the following quantity:~\cite{Coh1,Coh2} 	
	\begin{equation}\label{}
		\mathcal{C}(\rho) = \sum_{j \neq \ell} |\rho_{j\ell}| .
	\end{equation}
	In our case, if $\Lambda_t$ is CP-divisible, then	
	\begin{equation}\label{dC}
		\frac{\mathrm{d}}{\mathrm{d}t} \mathcal{C}(\Lambda_t(\rho)) \leq 0;
	\end{equation}
	besides, in the qubit case the above condition is also sufficient for CP-divisibility.	
	
	As for quantum regression, Prop.~\ref{prop:regression} has an immediate generalization to the qu$d$it case:
	\begin{Proposition}\label{prop:regression_multi}
		Given a Hamiltonian $\H$ in the form~\eqref{eq:h_multi} and an environment state $\rho_{\rm B}$, the couple $(\H,\rho_{\rm B})$ satisfies quantum regression if and only if, for all $t_n\geq t_{k-1}\geq\ldots\geq t_0\geq0$, the following equality holds:
		\begin{equation}\label{eq:condition_multi}
			\tr\,\Bigl[\e^{-\ii \Delta t_n H_{j_n}}\cdots\e^{-\ii \Delta t_0 H_{j_0}}\rho_{\rm B}\,\e^{\ii\Delta t_0H_{\ell_0}}\cdots\e^{\ii \Delta t_nH_{\ell_n}}\Bigr]=\prod_{k=0}^n\tr\,\bigl[\e^{-\ii\Delta t_kH_{j_k}}\rho_{\rm B}\,\e^{\ii\Delta t_k H_{\ell_k}}\bigr]
		\end{equation}
		for all $j_0,\ell_0,\ldots,j_n,\ell_n\in\{0,1,\dots,d-1\}$, where $\Delta t_k=t_k-t_{k-1}$ and $\Delta t_0=t_0$.
	\end{Proposition}
	\begin{proof}
		The proof follows by repeating the same steps as in Prop.~\ref{prop:regression} by just extending all sums from $0$ to $d-1$.
	\end{proof}
	In particular, a multilevel generalization of Corollary~\ref{coroll} holds as well: if $\{H_0,\dots,H_{d-1}\}$ is a \textit{commutative} family of self-adjoint operators satisfying quantum regression, then necessarily all dephasing functions satisfy $|\varphi_{j\ell}(t)|=1$, so that the dephasing of all matrix elements is trivial.
	
	\subsection{Dephasing-type generalized spin-boson models}
	
	We will finally introduce a family of generalized spin-boson (GSB) models whose reduced dynamics corresponds to a multilevel dephasing channel. Let $\hilbb$ be the symmetric Fock space associated with a continuous boson bath as before. We define
	\begin{equation}
		\mathbf{H}=H_{\rm S}\otimes\oper_{\rm B}+\oper_{\rm S}\otimes H_{\rm B}+\sum_{j=0}^{d-1}\ketbra{j}{j}\otimes\left(b(f_j)+b^\dag(f_j)\right),
	\end{equation}
	with $H_{\rm S}=\sum_j\omega_j\ketbra{j}{j}$, $H_{\rm B}$ as before, and where $f_0,\dots,f_{d-1}$ is a family of square-integrable coupling functions. The case previously studied is recovered by setting $d=2$ and $f_1=-f_0$. Clearly, this Hamiltonian can be decomposed as in Eq.~\eqref{eq:h_multi} with
	\begin{equation}
		H_j=\omega_j+H_{\rm B}+b(f_j)+b^\dag(f_j),
	\end{equation}
	again with domain $\mathcal{D}(H_j)=\mathcal{D}(H_{\rm B})$. The family of operators $\{H_0,\dots,H_{d-1}\}$ is noncommutative excepting the trivial case $f_0=f_1=...=f_{d-1}$.
	
	Correspondingly, by defining $\Lambda_t(\rho)=\tr_{\rm B}\left[\e^{-\ii t\mathbf{H}}\left(\rho\otimes\rho_{\rm B}\right)\e^{\ii t\mathbf{H}}\right]$ as usual, we obtain a multilevel dephasing channel with
	\begin{equation}
		\varphi_{j\ell}(t)=\Braket{\vac|\e^{\ii tH_{\ell}}\e^{-\ii tH_j}|\vac}.
	\end{equation}
	By assuming the same constraints as in Eq.~\eqref{eq:constraint} and using Weyl operators like in the qubit case, we obtain the following expression for the evolution group associated with each $H_j$:
	\begin{equation}\label{eq:evolution}
		\e^{-\ii tH_j}=\e^{-\ii t\tilde{\omega}_j}\W{\frac{f_j}{\omega}}\e^{-\ii tH_{\rm B}}\Wdag{\frac{f_j}{\omega}},
	\end{equation}
	where
	\begin{equation}
		\tilde{\omega}_j=\omega_j-\int\mathrm{d}\omega\:\frac{|f_j(\omega)|^2}{\omega},
	\end{equation}
	thus yielding an immediate generalization of Prop.~\ref{prop:dephasing}:
	\begin{Proposition}\label{prop:dephasing_multi}
		For all $j,\ell=0,\dots,d-1$, the dephasing function $\varphi_{j\ell}(t)$ in Eq.~\eqref{eq:dephasing_multi} is given by
		\begin{equation}\label{eq:dephasing_gsb}
			\varphi_{j\ell}(t) = \e^{-\ii(\omega_j-\omega_\ell)t}\e^{\ii\theta_{j\ell}(t)}\exp\left\{-\int\mathrm{d}\omega\;\left(|f_j(\omega)-f_\ell(\omega)|^2\right)\frac{1-\cos\omega t}{\omega^2}\right\}.
		\end{equation}
		where
		\begin{equation}\label{eq:theta}
			\theta_{j\ell}(t)=2\Im\int\mathrm{d}\omega\;f_\ell(\omega)^*f_j(\omega)\frac{1-\cos\omega t}{\omega^2}+\int\mathrm{d}\omega\,\left(|f_j(\omega)|^2-|f_\ell(\omega)|^2\right)\left(\frac{t}{\omega}-\frac{\sin\omega t}{\omega^2}\right)
		\end{equation}
	\end{Proposition}
	\begin{proof}
		Same calculations as the proof of Prop.~\ref{prop:dephasing}.
	\end{proof}
	When $f_j=-f_\ell$, in particular, $\theta_{j\ell}(t)=0$ and the result of Prop.~\ref{prop:dephasing} is recovered. Notice that, again, while Eq.~\eqref{eq:evolution} holds under the strict constraints~\eqref{eq:constraint} (which, in particular, constrain all coupling functions $f_j(\omega)$ to vanish sufficiently quickly at $\omega=0$), the final result makes sense for a much larger class of function: indeed, since
	\begin{equation}
		1-\cos\omega t\sim\frac{1}{2}\omega^2t^2,\qquad\sin\omega t\sim\omega t\qquad (\omega\to0),
	\end{equation}
	all integrands in Eqs.~\eqref{eq:dephasing_gsb}--\eqref{eq:theta} have no singularities near $\omega=0$ even if $f_j(0),f_\ell(0)$ are finite.
	
	However, while this discussion largely mirrors the one in the previous section, an important difference emerges: while the dephasing function $\varphi(t)$ given by Eq.~\eqref{eq:varphi} was proven to be well-defined for every continuous function $f(\omega)$ satisfying $f(\omega)=\mathcal{O}(1)$ at large values of $\omega$, in this more general case the second integral in Eq.~\eqref{eq:theta} requires the function $|f_j(\omega)|^2-|f_\ell(\omega)|^2$ to vanish at $|\omega|\to\infty$.
	
	With this comment in mind, we can again search for choices of the coupling functions yielding an exponential dephasing at all times. If $f_j(\omega)=-f_\ell(\omega)=\text{const.}$, either on all positive and negative energies ($-\infty<\omega<\infty$) or on positive energies ($0\leq\omega<\infty$), then the additional phase $\theta_{j\ell}(t)$ vanishes identically and, as discussed in the previous section, an exponential dephasing is obtained. This happens for all $j,\ell$ if and only if we can split all $d$ levels into two groups of $d-r$ and $r$ levels such that, up to a suitable permutation of the indices,			
	\begin{equation}\label{fff}
		f_0(\omega) = \ldots = f_r(\omega) = f \  , \qquad f_{r+1}(\omega) = \ldots = f_{d-1}(\omega) = - f  ,
	\end{equation}
	with $f$ being some constant. In this case, all additional phase factors $\theta_{j\ell}(t)=0$ and hence the dephasing matrix has the following structure: for $j,\ell\leq r$ and $j,\ell\geq r$,		
	\begin{equation}\label{D1}
		\varphi_{j\ell}(t) =\e^{-\ii(\omega_j -\omega_\ell)t}
	\end{equation}
	while, for $j \leq r$, $\ell > r$, and for $j  > r $, $\ell < r$			
	\begin{equation}\label{D2}
		\varphi_{j\ell}(t) =\e^{-\ii(\omega_j -\omega_\ell)t}\e^{-\gamma |t|/2} ,  \qquad \gamma = 8\pi |f|^2 ,
	\end{equation}
	In other words, we can decompose the Hilbert space of the system as $\hilb_{\rm S}=\hilb_r \oplus \hilb_{d-r}$, with both $\hilb_r$ and $\hilb_{d-r}$ being \textit{decoherence-free subspaces}~\cite{FREE}. The resulting dynamics is thus completely analogous to the qubit dephasing dynamics with flat form factor, with the two subspaces $\hilb_r$ and $\hilb_{d-r}$ replacing the two qubit levels. In the qubit case, $r=0$ and there is no room for nontrivial decoherence-free subspaces. Accordingly, Eq.~\eqref{fff} implies that the interaction Hamiltonian can be compactly written as			
	\begin{equation}\label{}
		H_{\rm SB} = \Sigma \otimes (b(f) + b^\dagger(f)) ,
	\end{equation}
	with			
	\begin{equation}\label{}
		\Sigma = \Pi_r - \Pi_{d-r} = \sum_{j=0}^{r} \ketbra{j}{j} - \sum_{j=r+1}^{d-1} \ketbra{j}{j} ,
	\end{equation}
	which reduces to $\sigma_z$ in the qubit case.
	
	Clearly, Prop.~\ref{prop:flat} has an immediate generalization to this case.			
	\begin{Proposition} \label{prop:flat_multi0}
		The dephasing-type GSB model with form factors as in Eq.~\eqref{fff} satisfies quantum regression.				
	\end{Proposition}
	\begin{proof}
		Same as in the proof of Prop.~\ref{prop:flat}.
	\end{proof}
	Since the case presented above, while multidimensional, is essentially analogous to the qubit one, we may wonder whether it is still possible to obtain a ``genuinely'' multidimensional dephasing which satisfies regression. Consider the most general case in which all $f_j(\omega)$ are constant-valued, either on the full real line or on the half-line. Clearly the first addend in Eq.~\eqref{eq:theta} vanishes as long as we require all functions to be real-valued; however, the second integral does \textit{not} converge in both cases, since the integrand is $\mathcal{O}(\omega^{-1})$ at large $|\omega|$.
	
	{Only} in the case $-\infty<\omega<\infty$, it is possible to recover a finite value via a cutoff procedure. Precisely, if we engineer the couplings $f_j(\omega)$ such that			
	\begin{equation}\label{f[]}
		f_j(\omega) = f_j \ \ \omega \in [-\omega_{\rm cut},\omega_{\rm cut}] , \ \ \ f_j \in \mathbb{R} ,
	\end{equation}
	and $f_j(\omega) =0$ otherwise, then $\theta_{j\ell}(t) = 0$, and hence the cutoff dependent dephasing reads
	\begin{equation}\label{lim}
		\varphi^{\omega_{\rm cut}}_{j\ell}(t) =  \e^{-\ii(\omega_j-\omega_\ell)t} \exp\left\{-  (f_j-f_\ell)^2 \int_{-\omega_{\rm cut}}^{\omega_{\rm cut}} \mathrm{d}\omega\;\frac{1-\cos\omega t}{\omega^2}\right\}.
	\end{equation}
	Finally, in the limit $\omega_{\rm cut}\to\infty$,
	\begin{equation}\label{}
		\varphi_{j\ell}(t) := \lim_{ \omega_{\rm cut} \to \infty}  \varphi^{\omega_{\rm cut}}_{j\ell}(t) = \e^{-\ii(\omega_j-\omega_\ell)t}\e^{-|\gamma_j - \gamma_\ell||t|/2} ,
	\end{equation}
	one obtains a semigroup dephasing. Moreover, again repeating the same computations as in the proof of Prop.~\ref{prop:flat}, one obtains the following result.
	\begin{Proposition} \label{prop:flat_multi} The dephasing-type GSB model with form factors as in Eqs.~\eqref{f[]}--\eqref{lim} satisfies quantum regression in the limit $\omega_{\rm cut}\to\infty$.
	\end{Proposition}
	We point out, however, that such a result is crucially dependent on the choice of cutoff, differently from what happens in the qubit case (Prop.~\ref{prop:flat}) as well as in its generalization based on decoherence-free subspaces (Prop.~\ref{prop:flat_multi0}). Choosing a different (asymmetric) cutoff procedure would yield an additional phase term which would spoil the validity of quantum regression, though still ensuring a weaker condition as the one in Eq.~\eqref{eq:condition-pd-weaker} for the qubit case.
	
	\section{Concluding remarks} \label{SEC-V}
	
	In this work we have provided an extensive comparative study of two common (and inequivalent) notions of Markovianity for open quantum systems, namely CP-divisibility and quantum regression, for a class of Hamiltonians whose reduced dynamics yields a purely dephasing channel. In doing so, we have first analyzed the two-dimensional (qubit) case, and eventually studied the general one.
	
	Necessary and sufficient conditions for CP-divisibility have been formulated in terms of the dephasing matrix characterizing such channels; remarkably, for these systems, CP-divisibility is equivalent to the (generally weaker) P-divisibility condition. Besides, a hierarchy of necessary and sufficient conditions for the validity of quantum regression has been found. Remarkably, for these models, the validity of the semigroup property at all times (and thus, \textit{a fortiori}, of CP-divisibility) is a necessary condition for quantum regression to hold; in fact, the semigroup property coincides with one of the conditions that the two-time correlations functions must satisfy. However, this condition is far from being sufficient for general systems: in fact, regression is violated whenever the ``blocks'' of the total system-bath Hamiltonian commute unless the dephasing is trivial, even if a semigroup dynamics is observed at all times.
	
	Working in parallel with the general case, we have also investigated in greater detail a family of dephasing-type generalized spin-boson (GSB) models, for which all conditions imposed by quantum regression can be recast in an explicit form by means of Weyl operators. In particular, a semigroup evolution is obtained in the (singular) limit case in which the coupling between the system and the boson environment is ``flat''. Remarkably, in all such cases we have shown that quantum regression is either exactly satisfied, or satisfied up to additional phase terms. We stress that such results are exact: no weak-coupling assumption is required.
	
	Incidentally, our analysis shows many analogies between the class of channels studied in this work and the multilevel amplitude-damping dynamics considered in Ref.~\cite{regression}, despite the two models describing strikingly different phenomena. In both cases, CP-divisibility and P-divisibility are equivalent; both channels can be realized by means of GSB models; in both cases, a semigroup evolution is obtained in the limiting case of a flat system-bath coupling, and, in such a case, quantum regression is also satisfied beyond the weak-coupling approximation. Such analogies leave room for many open questions. For one: is the equivalence between CP-divisibility and P-divisibility an accidental coincidence, or does such a correspondence hold for a larger family of quantum channels?
	
	In this direction, it is possibly worth noticing that, at a fixed time, both the (multilevel) amplitude-damping channel and the (multilevel) dephasing channel provide a homeomorphism between a group of matrices, respectively endowed with the usual row-column product and the Hadamard product, and a subset of the quantum channels; that is, in both cases, $\Pi_{A\star B}=\Pi_A\Pi_B$, with $\star$ being the corresponding product. This observation might be at the root of the analogies between the amplitude-damping and phase-damping channels, and may inspire a more thorough investigation of the class of channels satisfying similar properties.
	
	\section*{Acknowledgments}
	
	D.L. was partially supported by Istituto Nazionale di Fisica Nucleare (INFN) through the project “QUANTUM” and by the Italian National Group of Mathematical Physics (GNFM-INdAM); he also thanks the Institute of Physics at the Nicolaus University in Toru\'n for its  hospitality. D.C. was supported by the Polish National Science Center Project No. 2018/30/A/ST2/00837.
	
	\appendix
	
	\section*{Appendix: a complex integral}
	Let $a,b$ two real numbers. We shall compute the principal value integral
	\begin{equation}\label{eq:complexint}
		\pvint_{-\infty}^{\infty}\frac{\e^{\ii\omega a}-\e^{\ii\omega b}}{\omega^2}\,\mathrm{d}\omega.
	\end{equation}
	respectively for both $a,b\geq0$, or both $a,b\leq0$; the general case may be treated as well with analogous techniques.
	
	Suppose $a,b\geq0$. First of all, notice that the integrand can be immediately extended to a meromorphic function in the complex plane:
	\begin{equation}
		\zeta\in\mathbb{C}\setminus\{0\}\rightarrow f(\zeta)=\frac{\e^{\ii\zeta a}-\e^{\ii\zeta b}}{\zeta^2},
	\end{equation}
	with $\zeta=0$ being a simple pole for $f(\zeta)$, since $\e^{\ii \zeta a}-\e^{\ii\zeta b}\sim\ii(a-b)\zeta$ as $\zeta\to0$. Consequently, the residue of $f(\zeta)$ at the pole is
	\begin{equation}
		\mathrm{Res}_f(0)=\lim_{z\to0}f(\zeta)\zeta=\ii(a-b).
	\end{equation}
	The integral~\eqref{eq:complexint} can be obtained by circumventing the singularity at $\zeta=0$, and closing the integration contour in a semicircle, as shown in Fig.~\ref{fig:contour1}.
	
	\begin{figure}[ht]
		\centering
		\begin{subfigure}{0.5\linewidth}\centering
			\includegraphics[width=0.8\linewidth]{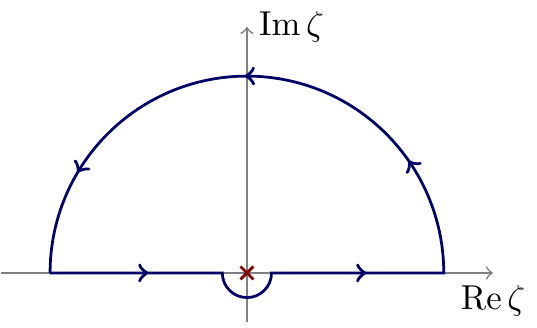}
			\caption{Case $a,b\geq0$.}\label{fig:contour1}
		\end{subfigure}%
		\begin{subfigure}{0.5\linewidth}\centering
			\includegraphics[width=0.8\linewidth]{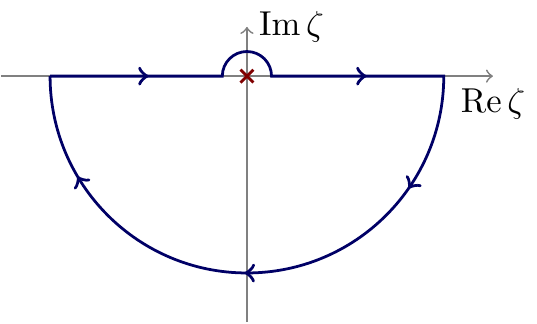}
			\caption{Case $a,b\leq0$.}\label{fig:contour2}
		\end{subfigure}
		\caption{Integration contour in the complex plane for the function $f(\zeta)$, depending on the signs of $a$ and $b$.}
		\label{fig:contour}
	\end{figure}
	By the residue theorem (and taking into account the contribution of the small semicircle enclosing the singularity), we obtain
	\begin{equation}\label{eq:complexint_result1}
		\pvint_{-\infty}^{\infty}\frac{\e^{\ii\omega a}-\e^{\ii\omega b}}{\omega^2}\,\mathrm{d}\omega=-\pi(a-b),\qquad a,b\geq0.
	\end{equation}
	Finally, the case $a,b\leq0$ can be either solved with the same technique as in the previous case, this time enclosing the path in the lower half-plane (see Fig.~\ref{fig:contour2}), or, equivalently, by reducing the problem to the previous case via the variable change $\omega\leftrightarrow-\omega$. This yields
	\begin{equation}\label{eq:complexint_result2}
		\pvint_{-\infty}^{\infty}\frac{\e^{\ii\omega a}-\e^{\ii\omega b}}{\omega^2}\,\mathrm{d}\omega=\pi(a-b),\qquad a,b\leq0.
	\end{equation}
	Finally, as a particular case, setting $a=0$ one obtains
	\begin{equation}\label{eq:complexint_result3}
		\pvint_{-\infty}^{\infty}\frac{1-\e^{\ii\omega b}}{\omega^2}\,\mathrm{d}\omega=\begin{cases}
			\pi b,&b\geq0;\\
			-\pi b,&b\leq0
		\end{cases}=\pi|b|.
	\end{equation}

\end{document}